\documentclass[a4paper,10pt,reqno]{amsart}
          \usepackage{amssymb}
	  \usepackage{amsmath}
          \usepackage{amsfonts}
          \usepackage[english]{babel}
          \usepackage[utf8]{inputenc}

\usepackage[margin=2.5cm]{geometry}

\usepackage{enumitem}

\usepackage{tikz}   
\usetikzlibrary{arrows,arrows.meta}  
\usepackage{caption}

 \usepackage[unicode,colorlinks,plainpages=false,hyperindex=true,bookmarksnumbered=true,bookmarksopen=false,pdfpagelabels]{hyperref}
 \hypersetup{urlcolor=cyan,linkcolor=blue,citecolor=red,colorlinks=true}
 \usepackage{color}
\newtheorem{thm}{Theorem}[section]
\newtheorem{cor}[thm]{Corollary}
\newtheorem{lem}[thm]{Lemma}
\newtheorem{prop}[thm]{Proposition}

\newtheorem{rem}[thm]{Remark}
\newtheorem{defn}[thm]{Definition}

\numberwithin{equation}{section}

\newcommand{\be}{\begin{equation}}
\newcommand{\ee}{\end{equation}}
\newcommand{\bea}{\begin{eqnarray}}
\newcommand{\eea}{\end{eqnarray}}
\newcommand{\ba}{\begin{aligned}}
\newcommand{\ea}{\end{aligned}}

\newcommand\br[1]{\{ #1 \}}                %

\begin{document}

{\title{Notes on the degenerate integrability of  reduced systems obtained from the master systems of free motion on
cotangent bundles of compact Lie groups}}\thanks{To appear in the proceedings of the XL Workshop on Geometric Methods in Physics, Bialowieza, 2-8 July, 2023.}

\maketitle

\begin{center}

L. Feh\'er${}^{a,b}$

\medskip
${}^a$Department of Theoretical Physics, University of Szeged\\
Tisza Lajos krt 84-86, H-6720 Szeged, Hungary\\
e-mail: lfeher@physx.u-szeged.hu

\medskip
${}^b$Institute for Particle and Nuclear Physics\\
Wigner Research Centre for Physics\\
 H-1525 Budapest, P.O.B.~49, Hungary

\end{center}

\begin{abstract}
The  reduction of the `master system' of free motion on the cotangent bundle $T^*G$  of a compact, connected and simply connected,
semisimple Lie group  is considered using the conjugation action of $G$.  It is  proved that the restriction of the reduced system
to the smooth component of the quotient space $T^*G/G$, given by the principal orbit type,
inherits the degenerate integrability of the master system.  The proof can be generalized easily to other interesting examples of Hamiltonian reduction.

\end{abstract}

 \setcounter{tocdepth}{2}

 \tableofcontents

 \vspace{1cm}

\def\1{{\boldsymbol 1}}                     %
\def\cH{{\mathcal H}}                       %
\def\tr{\mathrm{tr}}                        %
\def\ri{{\rm i}}                            %
\def\bC{\mathbb{C}}                         %
\def\bR{\mathbb{R}}                         %
\def\cF{{\mathcal F}}                       %
\def\reg{\mathrm{reg}}                      %
\def\dt {\left.\frac{d}{dt}\right|_{t=0}}   %
\def\cG{{\mathcal G}}                       %
\def\red{{\mathrm{red}}}                    %
\def\rank{{\mathrm{rank}}}                  %
\def\cM{\mathcal{M}}                        %
\def\fR{\mathfrak{R}}                       %
\def\cC{\mathcal{C}}                        %
\def\cO{\mathcal{O}}                        %
\def\fF{\mathfrak{F}}                       %
\def\fH{\mathfrak{H}}                       %
\def\fC{\mathfrak{C}}                       %
\def\pol{\mathrm{pol}}                      %
\def\fP{\mathfrak{P}}                       %
\def\pol{\mathrm{pol}}                      %
\def\princ{\mathrm{princ}}                  %
\def\cT{{\mathcal T}}                       %

\newpage

\section{Introduction}
\label{S:Intr}

In our lecture at the workshop we reported the results of   \cite{F1} on Poisson reductions of integrable `master systems' that
live on the classical doubles of a compact Lie group, $G$. The doubles in question are the cotangent bundle, and its
generalizations provided by the Heisenberg double \cite{STS}  and the internally fused quasi-Poisson/quasi-Hamiltonian double \cite{AKSM,AMM}.
The essence of the reduction consists in taking the quotient of the double by the pertinent `conjugation action' of $G$.
The simplest example of the master systems is the Hamiltonian system on $T^*G$
associated with free motion on $G$ along the geodesics of a
bi-invariant Riemannian metric. This Hamiltonian system, as well as its counterparts on the other doubles, enjoy
the property of degenerate integrability \cite{Nekh},  alternatively
 called non-commutative integrability \cite{MF} or superintegrability \cite{Fas,MPV}.

There exists a general principle \cite{J,Zu} stating that (degenerate) integrability is generically inherited under Hamiltonian reduction.
More concretely, this has been shown  by the studies \cite{Re1,Re2,Re2+,Re3,Re4} in which superintegrable spin Calogero--Moser--Sutherland type models
and their relativistic deformations  were derived by Hamiltonian reduction.
In the papers just cited,   integrability is stated to hold on some open  subset of the reduced phase space or on generic symplectic leaves.
Ideally, we would like to replace such `generic statements' by more specific statements about the subsets of the quotient space  where the integrability of the reduced system is valid.
We also would like to develop  mathematically rigorous proofs of the reduced integrability regarding
those examples for which we previously treated this issue by convincing, but admittedly not fully rigorous arguments \cite{F1}.

In this work, we make a modest step towards realizing the abovementioned goals by presenting  an exact proof
of the degenerate integrability of the reduction of the system of free motion on $T^*G$, focusing on the smooth part of the Poisson quotient
of the cotangent bundle that descends from the dense  open subset of principal orbit type for the conjugation action.
For readers who wish to go to the main results directly, we note that the full quotient space will be denoted
$\cM^\red =\cM/G$, where $\cM= T^*G$, and
$\cM_*^\red \subset \cM^\red$ denotes the dense open subset corresponding to the principal orbits.
Another interesting subset that appears in our investigation is the dense open subset $\cM_{**}^\red \subset \cM_{*}^\red$,
which corresponds to the principal orbits of $G$ in the `space of constants of motion'.
Our main results about degenerate integrability on these smooth Poisson manifolds are given by Theorems \ref{thm:45} and \ref{thm:411}.

Our claims will not surprise the experts, although we have not found these results
in the literature.  Our principal achievement is that  we managed to develop a method for analyzing
the specific example of the cotangent bundle that should be applicable to
many other examples. This is certainly the case regarding the master systems
studied in \cite{F1}, as  will be briefly indicated at the end of the paper.

The content of the text can be sketched as follows.
The next section is devoted to preliminaries, including the definitions of degenerate integrability on symplectic
and Poisson manifolds that we use. In Section 3, we  describe the unreduced master system of free motion on the cotangent bundle,
and explain its well-known integrability property.
After setting up the reduction, Section 4 consists of four subsections.
First, we prove the degenerate integrability on $\cM_{**}^\red$ in the sense Definition \ref{defn:23}
and establish very strong regularity properties of the
restricted reduced system.  Incidentally, these regularity properties (see e.g.~Lemma \ref{lem:43})
are precisely those that are assumed in \cite{LMV} for the construction of generalized action-angle coordinates.
Second, we prove that degenerate integrability in the sense of Definition \ref{defn:22} holds on every symplectic leaf of $\cM_{**}^\red$ of maximal
dimension.  Third, we elaborate the example of $\mathrm{SU}(2)$ to illustrate that $\cM_{**}^\red$ is a proper subset of
$\cM_{*}^\red$.  Fourth, we prove degenerate integrability on the full subset $\cM_*^\red \subset \cM$ of principal orbit type.
We conclude in Section 5 by giving an outlook on other applications of our method.
Finally, the  appendix contains a technical result which
plays a crucial role in Subsection 4.4.

\section{Preliminaries}
\label{S:P}

We first recall the notion of degenerate integrability, and then present some Lie theoretic
background material.
In the literature \cite{Fas,J,LMV,Nekh,Re2,Rud,Zu} , one can find slightly different, but essentially equivalent definitions of
degenerate integrability on symplectic  as well as on Poisson manifolds.
We shall use the following definitions.

\begin{defn}\label{defn:22}
Suppose that $\cM$ is a symplectic manifold of dimension $2m$ with associated Poisson bracket
$\br{-,-}$ and two distinguished subrings $\fH$ and $\fF$ of $C^\infty(\cM)$
satisfying the  following conditions:
\begin{enumerate}[itemsep=0pt]
\item{
The ring $\fH$ has functional dimension $r$ and $\fF$ has functional dimension $s$ such that
$r + s = \dim(\cM)$ and  $r<m$.}
\item{Both $\fH$ and $\fF$ form Poisson subalgebras of $C^\infty(\cM)$,  satisfying
$\fH\subset \fF$  and $ \{\cF, \cH\}=0$ for all $\cF\in \fF$, $\cH\in \fH$.}
\item{The Hamiltonian vector fields of the elements of $\fH$ are complete.}
\end{enumerate}
Then, $(\cM, \br{-,-}, \fH, \fF)$ is called a  degenerate integrable system of rank $r$.
The rings $\fH$ and $\fF$ are referred to as the ring of Hamiltonians and constants of motion, respectively.
\end{defn}

Recall that the functional dimension of a ring $\fR$ of functions on a manifold $\cM$ is $d$
if the exterior derivatives of the elements of $\fR$  \emph{generically} (that is on a dense open submanifold)
span a $d$-dimensional subspace of the cotangent space.
Condition $(3)$ above on the completeness of the flows is superfluous if the joint level surfaces
of the elements of $\fF$  are compact.
Degenerate integrability of a single Hamiltonian $\cH$ is  understood to mean
that there exist rings $\fH$ and $\fF$ with the above properties such that $\cH\in \fH$.
Observe that $\fF$ is either equal to or can
be enlarged to the centralizer of $\fH$ in the Poisson algebra $(C^\infty(\cM), \br{-,-})$.
In the literature the definition is often formulated in terms of functions
$f_1,\dots, f_r, f_{r+1},\dots f_s$ so that they generate $\fF$ and the first $r$ of them generate $\fH$.
If the definition is modified by setting $r=s=m$ and $\fH = \fF$, then one obtains
the notion of Liouville integrability.

\begin{defn}\label{defn:23}
Consider
 a Poisson manifold $(\cM,\br{-, -})$  whose Poisson tensor has maximal rank $2m\leq \dim(\cM)$
on a dense open subset.  Then,  $(\cM,\br{-,-},\fH,\fF)$ is called a degenerate integrable
system of rank $r$ if conditions (1), (2), (3) of Definition \ref{defn:22} hold, and
the Hamiltonian vector fields of the elements of $\fH$ span an $r$-dimensional subspace of
the tangent space over a dense open subset of $\cM$.
\end{defn}

The integrable systems of Definition \ref{defn:22} are integrable in the sense
of Definition \ref{defn:23}, too, since in the symplectic case the condition on the span of
the Hamiltonian vector fields of $\fH$ holds automatically.
Liouville integrability in the Poisson case results by imposing $r=m$
instead of $r<m$ in the definition.
In that case, our definition implies that $\fF$ is an abelian Poisson algebra
(in some papers this condition appears in the definition).

\medskip
Now turning to Lie theoretic background material,
let $G$ be a \emph{connected and simply connected} compact Lie group associated with a simple Lie algebra $\cG$.
Denote $\langle - , - \rangle_\cG$ the Killing form of $\cG$, possibly multiplied by a fixed real constant for convenience.
By definition, we call an element of $G$ \emph{regular} if its isotropy group  with respect to the conjugation action of $G$ on itself
is a maximal torus. Similarly, an element of $\cG$ is regular if its isotropy subgroup with respect
the adjoint action of $G$ on $\cG$ is a maximal torus of $G$.
The regular elements form the connected, dense open subsets $G_\reg \subset G$ and $\cG_\reg \subset \cG$.
For any maximal torus $T<G$ and its Lie algebra $\cT <\cG$, we put
$T_\reg:= T \cap G_\reg$ and $\cT_\reg:= \cT \cap \cG_\reg$.
 We often call the adjoint action of $G$ on $\cG$ also `conjugation action'
and denote it like that, as if $\cG$ was necessarily a matrix Lie algebra.

We recall \cite{Mich} that the functional dimensions of the
rings of conjugation invariant functions $C^\infty(\cG)^G$ and $C^\infty(G)^G$ are both
equal to the rank of the Lie algebra.
In fact, the differentials of the invariant functions span a $\rank(\cG)$-dimensional linear
space at every regular element (otherwise they span a lower dimensional space).

Let $T< G$ be a maximal torus, with Lie algebra $\cT < \cG$.
We shall utilize the fact that there always exists another maximal torus $T'$, with Lie algebra $\cT'$, satisfying
the relations
\be
T\cap T' = Z(G) \quad \hbox{and}\quad \cT \perp \cT',
\label{TT}\ee
where $Z(G)$ is the center of $G$ and the orthogonality  between $\cT$ and $\cT'$ refers to the Killing form.
In fact,  Kostant's  \cite{Ko}  maximal torus `in
apposition' to $T$ has this property.
An overview of this concept can be found  in \cite{Mein}.
For $G= \mathrm{SU(n)}$ and $T$ the standard maximal torus,
$T'$ can be presented as  the isotropy subgroup of the regular element $\Lambda_n\in \mathrm{SU(n)}$  defined by
\be
\Lambda_n = C \bigl(E_{n,1} + \sum_{k=1}^{n-1} E_{k,k+1} \bigr),
\label{SULa}\ee
where the $E_{i,j}$ denote the usual elementary matrices, and $C$ is a constant for which $(C)^n = (-1)^{n-1}$.
The matrix $C^{-1}\Lambda_n$  represents a cyclic permutation of the standard basis vector of $\bC^n$,
and the constant $C$ ensures $\det (\Lambda_n)=1$. In this case, it is elementary to check the properties
\eqref{TT}.

\begin{rem}
The simply connectedness of $G$ is not important, and $G$ can also be reductive.
However,
for a connected but not simply connected group, like $\mathrm{SO}(3)$, there could exist
exceptional regular orbits in $G$ for which the isotropy group is not connected and not abelian, but its identity component is abelian \cite{DK,GOV}.
For compact, connected Lie groups in general, the elements of $G$ whose isotropy subgroups with respect to the conjugation action are maximal tori
 fill the open, dense subset $G_\princ \subset G$.
In the simply connected case $G_\princ = G_{\reg}$ \cite{DK,GOV}, and we took this into account in  our definition of regular elements.
Our strong assumptions on the nature of the compact Lie group $G$ serve to `make life as simple as possible'.
\end{rem}

\section{The system of free motion on the cotangent bundle $T^*G$}
\label{S:Sec2}

 We identify $\cM:=T^*G$ with $G\times \cG^*$ by means of right-translations using the convention that
$\alpha_g \in T^*_g G$  is mapped to $(g, - R_g^*\alpha_g)\in G\times \cG^*$, where $R_g$ is right-multiplication by $g$.
We also identify the dual space $\cG^*$ with $\cG$  relying on the Killing form.
Then, the canonical Poisson bracket on the phase space
\be
\cM \equiv G \times \cG = \{(g,J)\mid g\in G,\, J\in \cG\}
\label{cM}\ee
can be written as
\be
\{ \cF, \cH\}(g,J) =
\langle \nabla_1 \cF, d_2 \cH\rangle_\cG - \langle \nabla_1 \cH, d_2 \cF \rangle_\cG + \langle J, [d_2 \cF, d_2 \cH]\rangle_\cG,
\label{PBcot}\ee
where the derivatives are taken at $(g,J)$.
Here, we use the $\cG$-valued derivatives of any $\cF \in C^\infty(\cM)$,  defined by
\be
 \langle X, \nabla_1 \cF(g,J) \rangle_\cG + \langle X', \nabla_1' \cF(g,J) \rangle_\cG := \dt \cF(e^{tX} g e^{tX'},J), \quad \forall
  (g,J)\in\cM, \, X,X' \in \cG,
 \label{nabG1}\ee
 and
 \be
 \langle X, d_2 \cF(g,J) \rangle_\cG  := \dt \cF(g, J + tX), \quad \forall
  (g,J)\in\cM, \, X \in \cG.
 \label{d2}\ee
The group $G$ acts by diagonal (simultaneous) conjugations of $g$ and $J$, i.e., the action of $\eta \in G$ on $\cM$
is furnished by the map
\be
A_\eta: (g,J) \mapsto (\eta g \eta^{-1}, \eta J \eta^{-1}).
\label{Aeta}\ee
This Hamiltonian action is generated by the moment map $\Phi: \cM \to \cG$,
\be
\Phi(g,J) = J - \tilde J
\quad \hbox{where}\quad \tilde J := g^{-1} J g.
\label{PhimomT}\ee
The space of $G$-invariant real functions, $C^\infty(\cM)^G$, forms a Poisson subalgebra of $C^\infty(M)$.
By definition, this is identified  as the Poisson algebra of smooth functions carried by the quotient space $\cM/G$.
Note in passing that the center $Z(G)<G$ acts trivially.

Below, and in the rest of the text,  we use the notation
\be
r:= \rank(\cG).
\label{rdef}\ee
Let us  consider the invariant Hamiltonians $\cH \in C^\infty(\cM)^G$ of the form
\be
\cH(g,J) = \varphi(J)
\quad\hbox{with}\quad \varphi\in C^\infty(\cG)^G.
\label{Hams1}\ee
That is,  $\cH = \pi_2^*(\varphi)$ using the natural projection $\pi_2: \cM \to \cG$,
\be
\pi_2(g,J):= J.
\ee
There are $r$ functionally independent Hamiltonians in this set,
since the ring of invariants for the adjoint action of $G$ on $\cG$,
$C^\infty(\cG)^G$, is freely generated by $r$ basic polynomial invariants.
The Hamiltonian vector field engendered by $\cH$ \eqref{Hams1} can be written as
\be
\dot{g} = (d \varphi(J)) g, \qquad \dot{J} =0,
\label{unredeq1}\ee
and its integral curve  through the initial value $(g(0), J(0))$ reads
\be
(g(t), J(t))= ( \exp(t d \varphi(J(0))) g(0), J(0)).
\label{unredsol1}\ee
The corresponding constants of motion are given by arbitrary functions of $J$ and $\tilde J$ \eqref{PhimomT}.

Now we introduce the map
\be
\Psi: G \times \cG \to \cG \times \cG,
\quad
\Psi: (g,J) \mapsto (g^{-1} J g, J).
\label{Psi}\ee
We shall need the  subset of the image
\be
\fC:= \Psi(\cM) \subset \cG \times \cG
\label{fC}\ee
defined by
\be
\fC_\reg:= \Psi(\pi_2^{-1}(\cG_\reg)) \subset \cG_\reg \times \cG_\reg.
\label{fCreg}\ee
Next, we record important properties of the map $\Psi$ and the set\footnote{The reason for working with  $\fC_\reg$ instead of $\fC$ below is that the latter
is not a smooth manifold in any natural way.}
 $\fC_\reg$.

\begin{prop}
The smooth map  $\Psi$ has the following properties:
\begin{enumerate}
\item{ $\Psi$ is constant along the flows \eqref{unredsol1} }.
\item{It is equivariant if $G$ acts on $\cG \times \cG$ by diagonal conjugations.}
\item{It is a Poisson map with respect to the product Poisson bracket
$-\br{-,-}_\cG \times \br{-,-}_\cG$ on $\cG \times \cG$, where
$\br{-,-}_\cG$ is the Lie--Poisson bracket on $\cG \simeq \cG^*$. }
\item{Its restriction on the dense open submanifold $\pi_2^{-1}(\cG_\reg)$ has constant rank,
equal to $\dim(\cM)-r$.}
\item{The image  $\fC_\reg$ \eqref{fCreg}
 is the joint zero set in $\cG_\reg \times \cG_\reg$ of the Casimir functions
$\chi_2(X,Y):= \chi(X) - \chi(Y)$  associated with the elements $\chi\in C^\infty(\cG)^G$.
Thus, it is an embedded submanifold of $\cG_\reg \times \cG_\reg$
of codimension $r$, and is also a Poisson submanifold of $\cG_\reg \times \cG_\reg$.}
\item{$\fC_\reg$ is connected.}
\end{enumerate}
\end{prop}
\begin{proof}
The claims follow from the formula \eqref{Psi} of the map $\Psi$ by straightforward inspection.
For example, one sees that the image of the derivative $D\Psi(g,J)$ is equal to
$T_{g^{-1}Jg} \cO_G(J) \oplus T_J \cG$,
where $\cO_G(J)$ is the $G$-orbit through $J$, and this implies property (4).
The subset $\fC_\reg \subset \cG_\reg \times \cG_\reg$ is an embedded submanifold since
 the differentials of the elements of $C^\infty(\cG)^G$ span an $r$-dimensional
 space at every $J\in \cG_\reg$, and its connectedness follows from
  the connectedness of $\pi_2^{-1}(\cG_\reg)= G \times \cG_\reg$.
\end{proof}

\begin{cor}
Consider the rings of functions
\be
\fH:= \pi_2^* (C^\infty(\cG)^G) \quad\hbox{and}\quad \fF:= \Psi^*(C^\infty(\cG \times \cG)).
\label{fHfF}\ee
The exterior derivatives of the
elements of $\fH$ and those of $\fF$ span, respectively, an $r$-dimensional and a $(\dim(\cM)-r)$-dimensional subspace
of the cotangent space at every point of $\pi_2^{-1}(\cG_\reg)$.
Consequently, these rings represent a degenerate integrable system on $\cM$ in the sense of Definition \ref{defn:22}.
\end{cor}
\begin{proof}
This is rather obvious, so it is enough to present a brief explanation.
The statement about the span of the exterior derivatives $ d\Psi^*(f)$, $f\in C^\infty(\cG \times \cG)$ follows since
\be
(d \Psi^*(f)) (p) = (D \Psi(p))^* df(\Psi(p)), \qquad \forall p\in \pi_2^{-1}(\cG_\reg),
\ee
and the dimension of the image of the  map $(D\Psi(p))^*$ is the same as the rank of $D\Psi(p)$, which
is equal to $(\dim(\cM) -r)$.
To see that $\fH \subset \fF$ holds,  notice that for $\cH(g,J) = \varphi(J)$, i.e., $\cH= \pi_2^*(\varphi)$, we also have
$\cH= \Psi^*(\tilde \varphi )$, where $\tilde \varphi(X,Y):= \varphi(Y)$.
\end{proof}

\begin{rem}
It is worth noting that
the restriction of $\Psi$ on $\pi_2^{-1}(\cG_\reg)$ yields a surjective submersion onto the manifold $\fC_\reg$.
Of course, $\fC_\reg$ can also be viewed as an embedded (but not closed) submanifold of $\cG \times \cG$.
\end{rem}

\section{Integrability on the Poisson quotient of the cotangent bundle}

We start by noting
 that there exist elements of  $\cM$, as well as elements of $\fC_\reg$, whose
 isotropy group with respect to the $G$-action is
given by the center $Z(G)$.
One can see this by direct inspection, or by using
a pair of maximal tori $T$ and $T'$ of $G$ that satisfy  the relations in \eqref{TT}.
Indeed, any pairs of elements of the form
\be
(g,J) \in T_\reg \times \cT'_\reg \subset \cM
\quad\hbox{and}\quad
 (\eta^{-1} J \eta, J) \in  \cT_\reg \times \cT'_\reg \subset \fC_\reg
 \ee
 possess the claimed properties.
Since $Z(G)$ is contained in every isotropy group, the elements having $Z(G)$ as isotropy group
form the dense open, connected subsets of principal isotropy type \cite{DK,GOV,Mich}.
In particular, denoting the isotropy group of $x$ by $G_x$, we have the dense open submanifold
\be
\cM_*:= \{ x \in \cM\mid G_x = Z(G)\}.
\ee
 The adjoint group $\bar G= G/Z(G)$ acts freely on $\cM_*$, and thus
\be
\cM^\red_* := \cM_*/G
\ee
becomes a connected,  smooth Poisson manifold, equipped with the Poisson algebra
\be
C^\infty(\cM^\red_*) \equiv C^\infty(\cM_*)^G.
\ee
It is clear that the flows of $\fH$ \eqref{fHfF} do not leave $\cM_*$.
Since
\be
\fH \subset C^\infty(\cM)^G,
\ee
the restrictions of the elements of $\fH$ descend  to an abelian Poisson algebra on $\cM_*^\red$.
We denote this as
\be
\fH_\red \subset C^\infty(\cM_*^\red).
\label{fHred}\ee
\emph{The functional dimensional of the ring of these reduced Hamiltonians
on $\cM_*^\red$ is  $r$, with \eqref{rdef}.}

In fact,
for any $\cH= \pi_2^* (\varphi)$ with  $\varphi \in C^\infty(\cG)^G$, we have the reduced Hamiltonian
$\cH_\red \in \fH_\red$, which satisfies
\be
\cH= \pi_2^* (\varphi) = \pi^* (\cH_\red),
\ee
where
\be
\pi: \cM_* \to \cM^\red_*
\label{piM*}\ee
 is the canonical projection.
Therefore
\be
d \cH =  \pi_2^* (d \varphi) =   \pi^* (d \cH_\red).
\ee
We know that $(D\pi(x))^*$ is injective for every $x \in \cM_*$. Hence, the linear span of the exterior derivatives
of the elements of $\fH_\red$ at $\pi(x)$ has the same dimension as the linear span of the exterior
derivatives of the elements of $\fH$ at $x$. But the latter is equal to $r$ if $x\in \pi_2^{-1}(\cG_\reg)$.

\begin{rem}
Recall that $\cM_*$ and $\pi_2^{-1}(\cG_\reg)$ are dense open subsets of $\cM$.
Consequently, $\cM_* \cap \pi_2^{-1}(\cG_\reg)$  is a dense open subset of $\cM_*$, and
is also a dense open subset of $\pi_2^{-1}(\cG_\reg)$.
Since we have a principal fiber  bundle, $\pi: \cM_* \to \cM^\red_*$, in particular $\pi$ is both open and continuous, we see
that
\be
(\cM_* \cap \pi_2^{-1}(\cG_\reg))/G \subset \cM^\red_*
\label{dense1}\ee
is dense and open.
\end{rem}

Denote by $\fF^G$ the subring of $G$-invariant elements of $\fF$ \eqref{fHfF}.
This is a Poisson subalgebra of $C^\infty(\cM)^G$, since $\fF$ is a Poisson subalgebra of $C^\infty(\cM)$.
The latter statement holds since $\fF = \Psi^* (C^\infty(\cG \times \cG))$ and $\Psi$ is a Poisson map.
In the end, we shall prove that the restrictions of the elements of $\fF^G \subset C^\infty(\cM)^G$ on $\cM_*$  give rise to a ring of
functional dimension $(\dim(\cM_*^\red)-r)$  on $\cM_*^\red$, which will show the integrability of the reduced
system on the Poisson manifold $\cM_*^\red$. But first we focus on a dense open subset $\cM_{**}^\red \subset \cM_*^\red$, defined below,
which is invariant under all the flows of our interest.

\subsection{Integrability on $\cM_{**}^\red$}

To start, let us introduce the set
\be
\fC_*:= \{ x \in \fC_\reg \mid G_x = Z(G)\} \subset \fC_\reg.
\label{fC*}\ee
This is the union of the principal $G$-orbits in $\fC_\reg$, and is thus a \emph{connected, dense open submanifold} of $\fC_\reg$.
Then, using the map $\Psi$ \eqref{Psi},  define
\be
\cM_{**}:=\Psi^{-1}(\fC_*).
\label{M**}\ee

We remarked before that the restricted map
\be
\Psi:  \pi_2^{-1}(\cG_\reg) \to \fC_\reg
\ee
is a \emph{smooth, surjective submersion}.
Now recall the standard facts that submersions are open maps and the inverse image of a dense set
by an open map is also dense.
This has the following important consequence.

\begin{lem}
The  inverse image $\cM_{**}$ \eqref{M**} of $\fC_*$ \eqref{fC*}
is a dense open subset of $\cM_*$.
 \end{lem}
 \begin{proof}
 We give  a proof that displays all details.
 To begin, let us note that we have the chain of dense open subsets
 \be
 \cM_{**} \subset \pi_2^{-1}(\cG_\reg) \subset \cM.
 \label{chain}\ee
 First, $\pi_2^{-1}(\cG_\reg) \subset \cM$ is dense open since $\pi_2: \cM \to \cG$ is a submersion.
 Second,  $\cM_{**}$  can be viewed as the inverse image of $\fC_* \subset \fC_\reg$
 under  the restricted map
 \be
 \Psi\vert_{\pi_2^{-1}(\cG_\reg)}: \pi_2^{-1}(\cG_\reg) \to \fC_\reg,
 \ee
 which is a surjective submersion, implying that  $\cM_{**} \subset \pi_2^{-1}(\cG_\reg)$ is a dense open subset.
 Next, we  observe that the inverse image \eqref{M**} is contained in $\cM_*$.
 This follows since the $G$-equivariance of $\Psi$  entails that $G_x < G_{\Psi(x)}$ for all $x\in \cM$.
Since $G_y = Z(G)$ for $y=\Psi(x)\in \fC_*$, and $Z(G)$ is contained in every isotropy group,  we
must have $G_x = Z(G)$ for $x\in \Psi^{-1}(\fC_*)$.

We see from  \eqref{chain} that $\cM_{**}$ is a dense open subset of $\cM$.
The claim of the lemma follows by combining this with the property that $\cM_{**}$ is contained in the
open  subset $\cM_*$ of $\cM$.
 \end{proof}

The construction ensures that
\be
\cM_{**}^\red:= \cM_{**}/G \quad \hbox{and}\quad
\fC_*^\red:= \fC_*/G
\ee
naturally become smooth Poisson manifolds. If $\psi$ denotes the restriction of $\Psi$ \eqref{Psi} on $\cM_{**}$, and
$p_1$ is the restriction of $\pi$ \eqref{piM*}, then we obtain the commutative diagram of Figure \ref{Diag-DIS}.
The fact that $\psi_\red$  in Figure  \ref{Diag-DIS} is a Poisson map follows from the  isomorphisms
\be
C^\infty(\cM_{**}^\red) \simeq C^\infty(\cM_{**})^G, \qquad C^\infty(\fC_*^\red) \simeq C^\infty(\fC_*)^G
\ee
by using that $\psi$ is a Poisson map and $\psi^*\left(C^\infty(\fC_*)^G\right) \subset C^\infty(\cM_{**})^G$, since $\psi$ is $G$-equivariant.
All the other properties mentioned in the text below the figure can be verified by simple diagram tracing.

\begin{figure}[ht]
\centering
  \captionsetup{width=.8\linewidth}
   \begin{tikzpicture}
 \node (A)  at (-1.7,1.2) {$\cM_{**}$};
 \node (B)  at (1.7,1.2) {$\fC_{*}$};
 \node (C)  at (-1.7,-1.2) {$\cM_{ * *}^{\red}$};
 \node (D)  at (1.7,-1.2) {$\fC_{*}^{\red}$};
  \path[->] (A) edge  node[above] {$\psi$}  (B); \path[->] (B) edge node[right] {$p_2$}  (D);
  \path[->] (A) edge node[left] {$p_1$}  (C); \path[->] (C) edge node[above] {$\psi_{\red}$}  (D);
 \end{tikzpicture}
 \caption{The ingredients of the proof of degenerate integrability on $\cM_{**}^\red$. All four  sets are smooth Poisson manifolds and all  maps are smooth, surjective
  Poisson submersions.
 $\fC_*$ contains those points of $\fC_\reg :=\Psi( \pi_2^{-1}(\cG_\reg)) \subset \fC := \Psi(\cM)$  whose isotropy group is the center
 $Z(G)$ of $G$.
  $\cM_{**}:= \Psi^{-1}(\fC_*)$
 with the $G$-equivariant map $\Psi:\cM \to \cG \times \cG$ that encodes the unreduced constants of motion.
  The map $\psi$ is the restriction of $\Psi$ \eqref{Psi},  $p_1$ and $p_2$ are principal bundle projections. The induced map $\psi_\red$ yields the constants of
   motion for the restricted reduced system \eqref{restred}.
   }
\label{Diag-DIS}
\end{figure}

\begin{lem}\label{lem:43}
The Hamiltonian vector fields generated by the elements of $\fH_\red$ \eqref{fHred} span
an $r$-dimensional subspace of the tangent space at every point of the dense open submanifold
\be
\cM_{**}^\red \subset \cM^\red_*.
\label{dense1+}\ee
\end{lem}
\begin{proof}
   Take an arbitrary point $x=(g,J) \in \cM_{**}$.
   At $y:=p_1(x) \in \cM_{**}^\red$, the values of the Hamiltonian vector fields generated by the elements of $\fH_\red$
   are obtained by applying the derivative $Dp_1(x)$ to the tangent vectors of the form
   \be
   (X g, 0) \in T_x \cM_{**}= T_g G \oplus T_J \cG,
   \label{vect1}\ee
   obtained from the original Hamiltonian vector fields \eqref{unredeq1}.
   The Lie algebra elements $X$ that appear here fill the isotropy subalgebra $\cG_J < \cG$, which is an abelian Lie algebra of dimension $r$ due to the regularity of $J$.
   We show that the subspace of these vectors has zero intersection with $\mathrm{Ker}(D p_1(x))$, which consists of the tangent vectors of the form
   \be
   ([Y,g], [Y,J])\in T_x\cM_{**}, \qquad \forall\, Y\in \cG.
  \label{vect2} \ee
   Indeed, if two tangent vectors having the respective forms \eqref{vect1} and \eqref{vect2} coincide, then so do their
   images in $T_{(\tilde J, J)} \fC_*$ obtained
   by the map $D\psi(x)$ (where $(\tilde J,J) =\psi(g,J)$).
    But the image of the tangent vector in \eqref{vect1}  is zero, while the image of the one in \eqref{vect2} is
   \be
   ([Y, \tilde J], [Y,J])\in T_{\psi(x)}\fC_* \subset  T_{\tilde J} \cG \oplus T_J\cG ,  \quad \hbox{with} \quad \tilde J = g^{-1} J g.
   \label{vect3}\ee
   By the definition of $\fC_*$, the  vector in \eqref{vect3} is equal to zero only for $Y=0$, in which case the vector \eqref{vect2} is also zero.
   Thus, we conclude that the $Dp_1(x)$ image of the vectors in \eqref{vect1} is $r$-dimensional.
\end{proof}

\begin{rem}\label{rem:44} Let us point out an interesting consequence of the proof of the previous lemma.
Take $x \in \pi_2^{-1}(\cG_\reg) \cap \cM_*$ and suppose that the linear span of the
reduced Hamiltonian vector fields arising from $\fH$ is $r_1 < r$ dimensional at $y= \pi(x) \in \cM_*^\red$.
Then, the Lie algebra of $G_{\Psi(x)}$ is at least $(r-r_1)$-dimensional.

The assumption means that $\mathrm{Ker}( D\pi(x))$ intersects the span of the Hamiltonian vectors field of $\fH$
at $x$ in an $(r-r_1)$-dimensional space. That is, there exist $X_1,\dots, X_{r-r_1}$ independent elements
of the Cartan subalgebra $\cG_J$ (for $x=(g,J)$), and corresponding elements $Y_1,\dots, Y_{r-r_1}$ of $\cG_J$, so that
\be
(X_i g, 0) = ([Y_i,g], [Y_i, J]), \qquad \forall i=1,\dots, r-r_1.
\label{fulleq}\ee
The left-hand side represents the values of Hamiltonian vector fields  that project to zero, and the right-hand side
gives infinitesimal gauge transformations. It is the second component of the equality \eqref{fulleq} that forces $Y_i \in \cG_J$, since $J$
is regular. The full equality implies that $[Y_i,\tilde J]=0$, i.e., $Y_i$ belongs to the Lie algebra of $G_{\Psi(x)}$, where $\Psi(x)= (\tilde J,J)\in \fC_\reg$.

We know that  $r$ reduced Hamiltonians are independent at every point of $(\pi_2^{-1}(\cG_\reg) \cap \cM_*)/G$,
because the restrictions of $\pi_2$ and $\pi$ are submersions on this set.
Our assumption means that $y$ is an equilibrium point for $(r-r_1)$ out of the $r$ independent reduced Hamiltonians.
 \end{rem}

Now we formulate our first main result.

\begin{thm}\label{thm:45}
Referring to Figure \ref{Diag-DIS} for the notations, consider the `restricted reduced system'
\be
(\cM_{**}^\red, \br{-,-}_{**}^\red, \fH_\red),
\label{restred}\ee
where $\br{-,-}_{**}^\red$ and $\fH_\red$ are the restrictions of the reduced Poisson brackets and Hamiltonians from $\cM_*^\red$ to
the dense open submanifold $\cM_{**}^\red$. Then, a ring of joint constants of motion for $\fH_\red$ is provided by
\be
\fF_\red:= \psi_\red^*( C^\infty(\fC_*^\red)).
\label{fFred}\ee
The reduced Hamiltonian vector fields associated with $\fH_\red$ span
an $r$-dimensional subspace (with $r=\rank(\cG)$) of the tangent space at every point of $\cM_{**}^\red$,
and the differentials of the elements of $\fF_\red$
span a codimension $r$ subspace of the cotangent space.
If $r>1$, then the quadruple $(\cM_{**}^\red,\br{-,-}^\red_{**}, \fH_\red, \fF_\red)$ gives a degenerate integrable system in the sense
of Definition \ref{defn:23}.  For $r=1$, which occurs for $G=\mathrm{SU(2)}$, the system is `only' Liouville integrable.
\end{thm}
\begin{proof}
Basically, all has already been said, but let us summarize the salient points.

 Consider a reduced Hamiltonian $\cH_\red\in \fH_\red$, defined by the relation
 \be
 \cH_\red \circ p_1 =  \cH\vert_{\cM_{**}}
 \quad\hbox{with some}\quad \cH\in \fH= \pi_2^*\left(C^\infty(\cG)^G\right). \ee
 The integral curves of $\cH_\red $ are of the form
 $p_1(x(t))$, where $x(t)$ is an integral curve of $\cH\in \fH$ in $\cM_{**}$.
 Figure \ref{Diag-DIS}  shows that
 \be
 \psi_\red (p_1(x(t))) = p_2 (\psi(x(t)))
 \ee
 is constant in $t$, since $\psi$ is constants along the original integral curve $x(t)$.
 It follows  that
 $\fF_\red$  \eqref{fFred} gives constants of motion.
We know that $\psi_\red$
 is a Poisson map,  and this implies that $\fF_\red$ forms a closed Poisson algebra.
  The dimension of the span of the differentials  of the elements of $\fF_\red$ at every point of $\cM_{**}^\red$ is equal to
   $\dim ( \fC_*^\red)$, because $\psi_\red$ is a  submersion. We have
   \be
   \dim ( \fC_*^\red) = \dim(\fC_*) - \dim(G) =( \dim(\cM_{**}) - r) - \dim(G) = \dim(\cM_{**}^\red) - r = \dim(G) -r. \ee
   It is obvious that $\fH_\red$ is contained in  the center of the Poisson algebra $\fF_\red$.

The differentials of the reduced Hamiltonians span an $r$-dimensional subspace of
   $T_y^* \cM_{**}^\red$ at every $y\in \cM_{**}^\red$, since their Hamiltonian vector fields span an $r$-dimensional subspace of $T_y \cM_{**}^\red$.
   Finally, we note that $r=\rank(\cG)$ is strictly smaller than half the maximal dimension of the symplectic leaves
   in $\cM_{**}^\red$,  if $r>1$.
    In fact, it follows from Lemma \ref{lem:denseleaves} given below that a dense open subset of $\cM_{**}^\red$ is filled by
   symplectic leaves of maximal dimension, which have codimension $r$.  If $r\neq 1$, then we see that
   \be
   r < \frac{1}{2} \left(\dim(G) - r\right) = \frac{1}{2}( \dim(\cM_{**}^\red) - r).
   \ee
   In the $r=1$ case we obtain equality instead of the above inequality, and then the restricted reduced system is `only'
   Liouville integrable.
   \end{proof}

\subsection{Integrability on the maximal symplectic leaves of $\cM_{**}^\red$}

Let us consider the restrictions of the moment map $\Phi$ \eqref{PhimomT},
\be
\Phi_*: \cM_* \to \cG
\quad\hbox{and}\quad
\Phi_{**}: \cM_{**} \to \cG.
\ee
These are submersions, because the action of $G/Z(G)$ on $\cM_*$, as well as on its $G$-stable submanifold $\cM_{**}$, is free.
Using again that submersions are open maps, we obtain that
\be
\Phi_*^{-1}(\cG_\reg)\subset \cM_*
\quad\hbox{and}\quad
\Phi_{**}^{-1}(\cG_\reg) \subset \cM_{**}
\ee
are dense open subsets.
According to the  theory of symplectic reduction \cite{OR}, the symplectic leaves of $\cM_*^\red$ and $\cM_{**}^\red$ are the
connected components of the reduced spaces
\be
\Phi_*^{-1}(\cO)/G
\quad \hbox{and}\quad
\Phi_{**}^{-1}(\cO)/G,
\label{orbred}\ee
where $\cO$ is such a $G$-orbit in $\cG$ whose respective inverse image is not empty.
See also  Remark \ref{rem:rem48} below.

\begin{lem}\label{lem:denseleaves}
The open dense subsets
\be
\Phi_*^{-1}(\cG_\reg)/G \subset\cM_*^\red
\quad\hbox{and}\quad
\Phi_{**}^{-1}(\cG_\reg)/G \subset \cM_{**}^\red
\label{regred}\ee
are unions of symplectic leaves of codimension  $r$, which are the connected components of the manifolds \eqref{orbred}
associated with the coadjoint orbits contained in $\cG_\reg$ (that have codimension  $r$ in $\cG^*\simeq \cG$).
\end{lem}
\begin{proof}
Let $C_i$ $(i=1,\dots, r)$ be a generating set of the $G$-invariant polynomials on $\cG$.
Recall that these functions are independent at every point of $\cG_\reg$.
The functions $C_i \circ \Phi$ are $G$-invariant, and descend to smooth functions $\cC_i^*$ on $\cM_*^\red$. We denote
the restricted functions on $\cM_{**}^\red$ as $\cC_i^{**}$.
It follows by using the submersion properties of $\Phi_*$ and the bundle projections $\pi$ and $p_1$, that each of the four sets of $r$-functions
\be
C_i \circ \Phi_*, \quad C_i \circ \Phi_{**}, \quad \cC_i^*, \qquad \cC_i^{**},\qquad i=1,\dots, r,
\label{4sets}\ee
are independent at every point of the corresponding subsets associated with $\cG_\reg$.
It is known  that
every orbit  $\cO\subset \cG$ is a joint level surface of the Casimir functions $C_i\in C^\infty(\cG)^G$.
The corresponding submanifolds $\Phi_*^{-1}(\cO)$ and $\Phi_*^{-1}(\cO)/G$ are joint
level surfaces of the functions $C_i \circ \Phi_*$ and $\cC_i^*$, respectively.
The analogous statements hold for $\Phi_{**}^{-1}(\cO)$ and $\Phi_{**}^{-1}(\cO)/G$ with the functions
$C_i \circ \Phi_{**}$ and $\cC_i^{**}$.
It is a straightforward consequence of the independence of the sets of functions $\{\cC_i^*\}$ and $\{\cC_i^{**}\}$
that the open dense subsets \eqref{regred} are composed of symplectic
leaves of codimension $r$.
\end{proof}

The moment map $\Phi$ can be represented as
\be
\Phi = \mu \circ \Psi
\quad\hbox{with}\quad
\mu: \cG \times \cG \to \cG
\quad\hbox{for which}\quad
\mu(X,Y):= Y -X.
\ee
Here, $\mu$ is the moment map for the diagonal conjugation action of $G$
on $\cG \times \cG$, equipped with the minus Lie--Poisson bracket on the first factor.
This implies that the Casimir functions $\cC_i^{**} \in C^\infty(\cM_{**}^\red)$ are elements of $\fF_\red$ \eqref{fFred}.
Note also that the reduced Hamiltonians  $\cH_i^\red \in C^\infty(\cM_{**}^\red)$ defined by
\be
\cH_i^\red \circ p_1  =  C_i \circ \pi_2
\ee
are elements of $\fH_\red \subset \fF_\red$.
It is easily seen (by using Lemma \ref{lem:43}) that the $2r$ functions
\be
\cC_1^{**},\dots, \cC_r^{**}, \cH_1^\red, \dots, \cH_r^\red
\label{2rfun}\ee
are independent at every point of $\Phi_{**}^{-1}(\cG_\reg) /G$.

\begin{prop}\label{prop:leaf}
If $r=\rank(\cG)>1$, then the restrictions of $\fH_\red$ and $\fF_\red$ to any symplectic leaf of codimension $r$ in $\cM_{**}^\red$
yield a degenerate integrable system in the sense of Definition \ref{defn:22}.
\end{prop}
\begin{proof}
Take  a regular orbit $\cO$ for which
\be
\Phi_{**}^{-1}(\cO)/G
\label{regleaf}\ee
is not empty. Choosing an arbitrary point $y \in \Phi_{**}^{-1}(\cO)/G$, we can find elements of $\fF_\red$  ,
say
\be
F_{1},\dots, F_{D}, \quad  D= \dim(\cM_{**}^\red) -3r
\ee
such that
\be
\cC_1^{**}, \dots, \cC^{**}_r, \cH_1^\red,\dots, \cH_r^\red, F_1, \dots, F_{D}
\label{2r+Dfun}\ee
are independent at $y$, This holds since there exist $\dim(\cM_{**}^\red) -r$ independent elements of $\fF_\red$
at every point of $\cM_{**}^\red$, and we know that the elements \eqref{2rfun} of $\fF_\red$ are independent at $y$.
According to standard results on manifolds, one can complement these functions to a local coordinate system
on some neighbourhood of $y$. In other words, there exists an open subset  $U\subset \cM_{**}^\red$ containing $y$, and $r$ smooth functions
$x_1,\dots, x_r$ on $U$, so that the restrictions of the functions \eqref{2r+Dfun} on $U$ together with the $x_i$ form local coordinates on $U$.
But then the restriction of the functions
 \be
 \cH_1^\red, \dots,\cH_r^\red, F_{1}, \dots, F_{D}, x_1,\dots, x_r
 \ee
 on $U \cap \Phi_{**}^{-1}(\cO)/G$ defines a local coordinate system on $\Phi_{**}^{-1}(\cO)/G$.
 In particular, the functional dimension of $\fF_\red$ restricted on $\Phi_{**}^{-1}(\cO)/G$
 is equal to the dimension of this symplectic manifold minus $r$.
 \end{proof}

\begin{rem}\label{rem:rem48}
Let us explain that  $\Phi^{-1}(\cO) \neq \emptyset$ for any orbit $\cO\subset \cG$, and
\be
\Phi^{-1}(\cO)\cap \cM_* \equiv \Phi^{-1}_*(\cO) \neq \emptyset,
\qquad \forall \cO \subset \cG_\reg.
\ee
For the proof, take two maximal tori $T$ and $T'$ subject to the relations \eqref{TT} and consider the equation
\be
\Phi(g,J) = J - g^{-1} J g = \zeta
\quad \hbox{with arbitrarily given}\quad g\in T_\reg, \, \zeta \in \cO\cap \cT',
\label{442}\ee
where the adjoint action is denoted by conjugation, as before.
 One sees from the orthogonality between $\cT$ and $\cT'$ \eqref{TT}  that  equation \eqref{442} admits a unique solution
for $J$. This proves the first statement, since every orbit intersects $\cT'$.
(We notice from Theorems 2.6 and 6.3 in \cite{Knop} that the full quotient spaces $\Phi^{-1}(\cO)/G \subset \cM/G$ are always connected.)
In the case of $\cO \subset \cG_\reg$, and thus $\zeta$ in $\cT'_\reg$,
we obtain that   $G_{(g,J)}  < T$ and $G_\zeta = T'$.
Using that $T\cap T' = Z(G)$ \eqref{TT}, this implies that $G_{(g,J)} = Z(G)$, since in general $G_{(g,J)} < G_{\Phi(g,J)}$. Thus the second statement follows.
For any orbit $\cO\subset \cG$, we have
\be
\dim (\Phi_{*}^{-1}(\cO)/G) = \dim(\cO),
\label{genorb}\ee
and also $\dim (\Phi_{**}^{-1}(\cO)/G) = \dim(\cO)$, if  $\Phi_{*}^{-1}(\cO) \neq \emptyset$ and $\Phi_{**}^{-1}(\cO) \neq \emptyset$.
One can determine the dimensions, for example,  by identifying these manifolds with reduced phase spaces arising from corresponding
Marsden--Weinstein type point reductions \cite{OR}.
We conjecture that every maximal symplectic leaf of $\cM_*^\red$ intersects $\cM_{**}^\red$ and the intersection
 gives a dense subset of the leaf in question.
This issue, and also the question if the same holds for arbitrary symplectic leaves of $\cM_*^\red$, deserves further study.
\end{rem}

\subsection{The example of $G=\mathrm{SU}(2)$}

In the $r=1$ case, i.e. for $G=\mathrm{SU}(2)$,
 $\fF_\red$ \eqref{fFred} is generated by $\cC_1^{**}$ and $\cH_1^\red$ that
stem from the Casimir function $C_1(X) = \langle X,X \rangle_\cG$.
The symplectic leaves in $\cM_{**}^\red$ are two dimensional, and support Liouville integrable systems.
Concretely,  these represent the center of mass version of the $2$-particle trigonometric
Sutherland model (given by the Hamiltonian \eqref{Suth} below) at different values of the coupling constant, and
their equilibrium points  belong to $\cM_*^\red \setminus \cM_{**}^\red$.

Now we explain the above statements. To start,
observe that in this case all elements of $G$ are regular, except the elements of $Z(G) = \{ \1_2, -\1_2\}$,
and also all elements of $\cG$ are regular except the zero element.
Note that  $\Phi^{-1}(0)$ never   intersects $\cM_{**}$.

Let us restrict to the dense open submanifold $\Phi^{-1}(\cG_\reg)$, which simply means
that we exclude the zero value of the moment map.
It is easily seen that all $G$-orbits in $\Phi^{-1}(\cG_\reg)$ admit a \emph{unique} representative
$(g,J)$ for which we have
\be
g = \mathrm{diag} (e^{\ri q}, e^{-\ri q}), \quad\hbox{with some}\quad  0<q < \pi
\ee
and
\be
J = \ri p (E_{11} - E_{22}) + \ri x \left( \frac{1}{(e^{2\ri q} -1)} E_{12} + \frac{1}{(e^{-2\ri q} -1)} E_{21} \right),
\ee
where $p$ is an arbitrary real number, and $x>0$ is a constant that parametrizes the non-zero coadjoint orbits of $G$.
The corresponding moment map value reads
\be
\Phi(g,J) = \ri x (E_{12} + E_{21}).
\ee
One can also calculate that
\be
 \tilde J=g^{-1} J g= \ri p (E_{11} - E_{22}) + \ri x \left( \frac{e^{-2\ri q}}{(e^{2\ri q} -1)} E_{12} + \frac{e^{2\ri q}}{(e^{-2\ri q} -1)} E_{21} \right).
 \ee
The isotropy groups are as follows:
\be
G_{(g,J)} = Z(G), \qquad \forall q, p,
\ee
and
\be
G_{(\tilde J, J)} = Z(G) \qquad \hbox{if}\quad (p,q)\neq (0, \frac{\pi}{2}).
\ee
In the case $(p,q) = (0, \frac{\pi}{2})$ we see that $J$ and $\tilde J$ are proportional to each other, and then
\be
G_{(\tilde J, J)} = G_J,
\ee
which is a maximal torus.

In our case $\fH_\red$ is generated by the single Hamiltonian
\be
\cH_\red = - \frac{1}{4} \tr(J^2) = \frac{1}{2} p^2 + \frac{1}{8} \frac{x^2}{\sin^2q}.
\label{Suth}\ee
It turns out that $q$ and $p$ are  Darboux variables.
The special value $(p,q)= (0, \frac{\pi}{2})$ corresponds to the global minimum of $\cH_\red$.
In particular, we observe that $\cM_*^\red \setminus \cM_{**}^\red$ is not empty, and contains the `interesting points'
regarding the reduced dynamics. Remember also Remark \ref{rem:44}.

It is instructive to check  that the value of the original Hamiltonian vector field on $\cM_*$
taken at the exceptional points $(g,J)\in (\cM_*\setminus \cM_{**})$  parametrized by $(p,q)=(0,\frac{\pi}{2})$
projects to zero on the reduced phase space.

\subsection{Integrability on $\cM_*^\red$}

Consider the ring of $G$-invariant polynomials on $\cG\times \cG$, denoted $\bR[\cG \times \cG]^G$,
which is a finitely generated polynomial ring, and at the same time is
a Poisson subalgebra of $C^\infty(\cG \times \cG)^G$.
As a consequence of general results  described in the appendix, the linear
span of the differentials of the invariant polynomials is \emph{the same},
 at every
point of $\cG \times \cG$,
as the linear span of the differentials of the invariant $C^\infty$ functions.
Now, we explain that this also holds after restriction on $\fC_*$.

\begin{lem}\label{lem:spanpol}
The dimension of the span of the differentials of the restrictions of the elements of $\bR[\cG \times \cG]^G$ on $\fC_*$ \eqref{fC*}
is equal to $\dim(G)-r$ at every point of $\fC_*$.
\end{lem}
\begin{proof}
Let us pick an arbitrary point $z\in \fC_*$. Then, we can choose $P_i \in \bR[\cG \times \cG]^G$, $i=1,\dots, \dim(G)$ for which
the differentials $dP_i(z)$ are linearly independent. For the justification of this statement, see the appendix.
Next, we can find further $\dim(G)$ functions $\xi_i \in C^\infty(\cG \times \cG)$ such that together with the $P_i$ they
give local coordinates on an open subset $U$ of $\cG\times \cG$ around $z$. We may require   that all elements of $U$ have $Z(G)$ as their isotropy group.
Moreover, the first $r$ of the $P_i$ can be taken to be the polynomials
\be
P_i(X,Y):= C_i(X) - C_i(Y) \quad\hbox{for}\quad i=1,\dots, r,
\ee
where the $C_i$ ($i=1,\dots, r)$ generate $\bR[\cG]^G$.
It follows that $U\cap \fC_*$ is the zero level set of the coordinate functions $P_1,\dots, P_r$ on $U$,
and thus the restricted functions
\be
{P_\alpha}\vert_{U\cap \fC_*},\quad {\xi_i}\vert_{U\cap \fC_*},\quad \alpha=r+1,\dots,\dim(G),\quad i=1,\dots, \dim(G),
\ee
form a local coordinate system on $U\cap \fC_*$.
This implies the claim immediately.
\end{proof}

\begin{rem}
Using pull-back by means of the natural embedding $\iota: \fC_* \to \cG \times \cG$, we have
\be
\iota^*\!\left( \bR[\cG\times \cG]^G\right)  \subset \iota^*\!\left(C^\infty(\cG \times \cG)^G\right) \subset C^\infty(\fC_*)^G .
\ee
One sees from Lemma \ref{lem:spanpol} and the lemmas in the appendix that the differentials of the elements of
any  of these  three rings of functions span
the same subspace of the cotangent space at every point of $\fC_*$.
\end{rem}

Lemma \ref{lem:spanpol} shows that the subalgebra of $\fF_\red$ \eqref{fFred} arising from the invariant polynomials
also guarantees the degenerate integrability of the reduced system on $\cM_{**}^\red$.
A crucial advantage of the polynomial invariants is that they give rise
to smooth functions on the full quotient space $\cM^\red$.
In this way, we obtain our second main result.

\begin{thm}\label{thm:411}
Consider the Poisson subalgebras of $C^\infty(\cM)^G$ given by
the pull-backs of the polynomial invariants
\be
 \pi_2^*(\bR[\cG]^G) \quad\hbox{and}\quad \Psi^*( \bR[\cG\times \cG]^G).
\ee
Their restrictions on the dense
open subset $\cM_* \subset \cM$ of principal orbit type descend to
Poisson subalgebras of $C^\infty(\cM_*^\red)$, denoted $\fH^\pol_\red$ and $\fF^\pol_\red$, respectively.
Then,  the quadruple $\left( \cM_*^\red, \br{-,-}_*^\red, \fH_\red^\pol, \fF_\red^\pol\right)$ represents
a degenerate integrable system in the sense of Definition \ref{defn:23}, provided that $r=\rank(\cG)>1$.
In the $r=1$ case, this quadruple provides a Liouville integrable system.
\end{thm}
\begin{proof}
It follows from Lemma \ref{lem:spanpol} that
the restrictions of $\fH^\pol_\red$ and $\fF^\pol_\red$ on
the \emph{dense open} subset $\cM_{**}^\red \subset \cM_*^\red$ have the same properties as the
corresponding rings of functions $\fH_\red$ and $\fF_\red$ that feature
in the
restricted reduced system described in  Theorem \ref{thm:45}. This proves  the claim immediately.
\end{proof}

\begin{rem}
The construction of $( \cM_*^\red, \br{-,-}_*^\red, \fH_\red^\pol, \fF_\red^\pol )$ entails
that this is a real-analytic integrable system.
After restriction on $\cM_{**}^\red\subset \cM_*^\red$, we see that the analogue of Proposition \ref{prop:leaf} remains valid
using $\fH_\red^\pol$ and $\fF_\red^\pol$.
An important consequence is that degenerate integrability
holds on every
such maximal  symplectic leaf of $\cM_*^\red$ that has nonempty intersection with $\cM_{**}^\red$.
(Here, it is assumed that $r= \rank(G)>1$.)
The statement follows by using that on a connected  real-analytic manifold the independence of functions
(and vector fields) at a single point implies their independence on a dense open subset.
See also the conjecture formulated at the end of Remark \ref{rem:rem48}.
\end{rem}

\section{Discussion}

The modest results that we presented took us a long time to develop, but our method is now
applicable  to treat the reductions of other integrable `master systems'
as well.  For example, it can be applied to the generalizations of the systems of
free motion \cite{F1} that live on the Heisenberg double and on the quasi-Poisson double of the compact Lie group $G$
(and also to the corresponding holomorphic systems).
Without going into any details, we recall that as a manifold the Heisenberg double can be identified
with
\be
\cM_2 := G \times \fP,
\label{cM2}\ee
where $\fP:= \exp(\ri \cG) \subset G^\bC$ is a closed submanifold of the complexification
$G^\bC$ of the compact Lie group $G$, and the internally fused quasi-Poisson double is equal to
\be
\cM_3 = G \times G
\ee
as a manifold. Then,  abelian Poisson algebras of `free Hamiltonians' are provided, respectively, by
the  rings of invariants
\be
\fH_2 = \pi_2^*(C^\infty(\fP)^G)
\quad \hbox{and}\quad
\fH_3 = \pi_2^*(C^\infty(G)^G),
\ee
where $\pi_2$ is the projection onto the second factor of the relevant Cartesian product.
The corresponding algebras of constants of motion can be described in a  similar manner as in the cotangent bundle case.
Namely \cite{F1}, they are encoded by the maps
\be
\Psi_2: \cM_2 \to \fP \times \fP
\quad\hbox{and}\quad
\Psi_3: \cM_3 \to G \times G,
\ee
which operate according to the  formulas
\be
\Psi_2: (g,L) \mapsto (g^{-1} L g, L)
\quad \hbox{and}\quad
\Psi_3: (A,B) \mapsto (A B A^{-1}, B).
\ee
These are Poisson and, respectively, quasi-Poisson maps with respect to suitable (degenerate) Poisson and quasi-Poisson structures
on the target spaces $\fP \times \fP$ and $G\times G$.   It should be clear from the very close similarity
to the map $\Psi: \cM \to \cG \times \cG$ \eqref{Psi} how the presented method works in these two cases.

Our method can be used to strengthen the results of \cite{FF}, where the useful restriction to
 a subset like $\fC_*$ \eqref{fC*} of the space of constants of motion already appeared. The method is also applicable to
  the reductions of the other master systems on $\cM$ \eqref{cM} and $\cM_2$ \eqref{cM2} that are built on the abelian Poisson algebras
defined by $\pi_1^* (C^\infty(G)^G)$.

 Our original motivation for the current work was that we wanted to better understand the results of Reshetikhin \cite{Re1,Re2,Re2+} about
  degenerate integrability on generic symplectic leaves of $T^*G/G$.
In the paper \cite{Re1}  it was explained that the reduced Hamiltonian descending from the kinetic energy can be interpreted
as describing a spin extension of the trigonometric Calogero--Sutherland model, if $\rank(G)>1$,
and the degenerate integrability of the `dual system' (a spin Ruijsenaars model)  obtained  from the reduction of
the master system based on $\pi_1^*(C^\infty(G)^G)$
was also demonstrated.
This work inspired our considerations. Its results are consistent with the claims of our Theorems \ref{thm:45} and \ref{thm:411}, which focus on degenerate integrability
in the sense of Definition \ref{defn:23}.
In the future, we wish to further explore the structure  of the interesting integrable systems
that arise from the pertinent reductions.

At this stage, the most challenging open problem is to understand the integrability properties on
arbitrary symplectic leaves of the smooth part $\cM_*^\red$ and on the
`singular part' $\cM^\red \setminus \cM_*^\red$ of the
full stratified Poisson space $\cM^\red$.
To improve our intuition, it  should be helpful to first analyze the
relatively low dimensional case of $G= {\mathrm{SU}}(3)$ in detail.

\bigskip
\medskip
\subsection*{Acknowledgement}
I wish to thank Alexey Bolsinov for  discussions, and especially
for a suggestion that led to the proof of Theorem \ref{thm:411}.
I also wish to thank  Anton Izosimov for kind help with the proof of Lemma \ref{lem:app1}.
This work was supported in part by the NKFIH research grant K134946.
\medskip

\appendix

\section{Useful technical results on invariant functions}

Let $G$ be a compact  Lie group acting on a $C^\infty$ connected manifold $M$,
and denote by $M_\princ \subset M$ the open, dense subset of principal orbit type for the $G$-action.
Informally speaking, this means that the points of $M_\princ$ have the smallest possible
isotropy group for the $G$-action. It is well-known that $M_\princ/G$ is a connected, smooth manifold, such that
the natural projection
\be
p: M_\princ \to M_\princ/G
\label{B1}\ee
is a smooth submersion.

\begin{lem}\label{lem:app1}
At any $x\in M_\princ$, the span of the differentials of the elements of $C^\infty(M)^G$ is a subspace of $T_x^*M$ whose
dimension is equal to the codimension in $M$  of the $G$-orbit through $x$.
\end{lem}
\begin{proof}
Let us recall \cite{DK,GOV,Mich}  that $M_\princ$ can be viewed as a locally trivial fiber bundle with fiber $G/G_x$ and base $M_\princ/G$.
Fix a coordinate system around $y= p(x)\in M_\princ/G$ and in terms of those coordinates consider an open ball $B_r(y)$ over which
the bundle is trivial, i.e., we have a $G$-equivariant diffeomorphism
\be
\sigma: B_r(y) \times G/G_x \to p^{-1}(B_r(y)),
\ee
where $G$ acts trivially on the  factor $B_r(y)$. Then we make a construction in four steps.

First, with $k:=\dim(M_\princ/G)$, for each $i=1,\dots, k$,
let us pick a function $\chi_i \in C^\infty(B_r(y))$ that coincides with the $i^{\mathrm{th}}$ coordinate on $B_{r/2}(y)$, and it vanishes outside
$B_{2 r/3}(y)$. It is standard that such functions exist.

Second, let
\be
f_i \in  C^\infty(B_r(y) \times G/G_x)
\ee
be the functions $\pi_1^*(\chi_i)$ where $\pi_1: B_r(y) \times G/G_x  \to B_r(y)$ is the obvious projection.

Third,
 we introduce
\be
F_i := f_i  \circ \sigma^{-1} \in C^\infty( p^{-1}(B_r(y)))^G.
\ee

Fourth, we let $\cF_i$ be the function on the full manifold $M$ that agrees with the $F_i$ on $p^{-1}(B_r(y))$ and vanishes identically
on the complement of this set.
The definition guarantees  that $\cF_i \in C^\infty(M)^G$.

It is obvious from the construction that the differentials
\be
d \cF_1(z),\dots,  d\cF_k(z)
\ee
are independent at every point $z\in p^{-1}(B_{r/2}(y))$, and the differential of every $G$-invariant function is in their span.
In particular, this holds  at the point $x\in M_\princ$ that we started with, whereby the proof is complete.
\end{proof}

\begin{lem}\label{lem:app2}
Let a compact Lie group $G$ act on  a finite dimensional real vector space $V$, and consider
 the ring of invariant
polynomials  $\bR[V]^G \subset C^\infty(V)^G$.
Then, at any $v\in V$, the span of the differentials of the elements of $\bR[V]^G$  is the same
as the span of the differentials of the elements of $C^\infty(V)^G$.
\end{lem}
\begin{proof}
Let $P_1,\dots, P_N$ denote a generating set of $\bR[V]^G$.
According to Schwarz's theorem \cite{Sch}, every smooth function $F\in C^\infty(V)^G$ can be represented in the form
\be
F = \cF \circ (P_1,\dots, P_N)
\ee
with a suitable function $\cF \in C^\infty(\bR^N)$.
Since
\be
dF(v) = \sum_{i=1}^N (\partial_i \cF)(P_1(v),\dots, P_N(v))  dP_i(v), \qquad \forall v\in V,
\ee
the claim follows.
\end{proof}

\end{document}